\documentclass[runningheads]{llncs}

\usepackage{graphicx}
\usepackage{color}
\usepackage{xcolor}
\usepackage{amsfonts}
\usepackage{amsmath}
\usepackage[mathscr]{euscript}
\usepackage{algorithm} 
\usepackage{algorithmicx}
\usepackage[noend]{algpseudocode}
\algrenewcommand\algorithmicrequire{\textbf{Input:}}
\algrenewcommand\algorithmicensure{\textbf{Output:}}
\usepackage{hyperref}

\definecolor{orcidgreen}{RGB}{166,206,57}
\newcommand{\orcidicon}{\includegraphics[width=0.20cm]{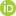}}
\def\orcidID#1{\renewcommand{\thefootnote}{\roman{footnote}}
  \unskip$^{\orcidicon}$\footnote{\orcidicon\color{orcidgreen}\,\scriptsize #1\color{black}\vspace*{-2pt}}\unskip
  \renewcommand*{\thefootnote}{\arabic{footnote}}\unskip
}
\DeclareMathOperator{\parent}{par}
\DeclareMathOperator{\child}{chld}

\DeclareMathOperator{\LCA}{LCA}
\DeclareMathOperator{\lca}{lca}
\DeclareMathOperator{\indeg}{indeg}
\newcommand{\coloneqq}{\mathrel{:=}}
\newcommand{\GG}{\ensuremath{\mathscr{G}}}

\begin{document}

\author{%
  \setcounter{footnote}{0}
  Patricia A.\ Ebert\inst{1}
  \fnmsep\orcidID{0009-0009-3361-7000}\and
  Peter F.\ Stadler\inst{2-6}
  \fnmsep\orcidID{0000-0002-5016-5191}\and 
  Marc Hellmuth\inst{7}
  \fnmsep\orcidID{0000-0002-1620-5508}
}

\institute{Department of Mathematics, Faculty of Science,
  Stockholm University, SE-106 91 Stockholm, Sweden,
  \email{patricia.ebert@math.su.se} 
  \and 
  Bioinformatics Group, Department of Computer Science, and
  Interdisciplinary Center for Bioinformatics, Leipzig University,
  H{\"a}rtelstrasse 16-18, D-04107 Leipzig, Germany,
  \email{studla@bioinf.uni-leipzig.de}
  \and 
  Institute for Theoretical Chemistry, University of Vienna, Vienna,
  Austria
  \and
  Max Planck
  Institute for Mathematics in the Sciences, Leipzig, Germany
  \and
  Facultad de Ciencias, Universidad Nacional de Colombia,
  Bogot{\'a}, Colombia
  \and
  Santa Fe Institute, Santa Fe, NM
  \and
  Theoretical Computer Science Group, Faculty of Mathematics and
  Computer Science, Leipzig University,
  Augustusplatz 10, D-04109 Leipzig, Germany,
  \email{marc.hellmuth@uni-leipzig.de}
}

\title{Best Matches in Phylogenetic Networks}

\authorrunning{Ebert, Stadler \& Hellmuth}
\titlerunning{Best Matches in Phylogenetic Networks}

\maketitle              % typeset the header of the contribution
\begin{abstract}
  Best match graphs (BMGs) were introduced in mathematical phylogenetics to
  describe the concept of closest relatives for related genes (leaves of
  rooted tree) in different organisms (defining leaf colors). We generalize
  this concept here to leaf-colored rooted networks, where least common
  ancestors are in general neither unique nor comparable. We characterize BMGs 
  of rooted networks as those vertex-colored digraphs that are
  properly colored and satisfy an easy-to-check condition that we call the
  sicor-in-hub property. BMGs can be recognized in linear time and an
  explaining network can be constructed in quadratic time. Analogous
  results are obtained for reciprocal best match graphs (RBMGs),
  where an edge $\{x,y\}$ corresponds to pairs of vertices with different
  color that are mutually closest relatives.
\end{abstract}

\section{Introduction}

The best match relation in leaf-colored rooted trees formalizes the idea of
closest relatives among genes in evolutionary biology. In this setting, the
rooted tree $T$ describes the phylogenetic relationships among a set $X$ of
genes represented by the leaves of $T$, while the inner vertices of $T$
represent ancestors of these genes. Then the least common ancestor
$\lca(x,y)$ of two genes $x,y \in X$ is an ancestor of both $x$ and $y$
such that there exists no vertex below it in $T$ with the same property.
Colors identify the species $\sigma(x)$ in which gene $x\in X$ is found. A
gene $y$ in species $\sigma(y)$ is then a \emph{best match} of a gene $x$
in species $\sigma(x)\ne\sigma(y)$ if there is no other gene $y'$ in the
same species as $y$, i.e., with $\sigma(y')=\sigma(y)$, such that the least
common ancestors satisfy $\lca(x,y')\prec\lca(x,y)$, where $\prec$ is the
ancestor partial order in $T$. The best match relation derived from the
leaf-colored tree gives rise to the vertex-colored \emph{best match graph}
$\GG(T,\sigma)$ with vertex set $X$ and an arc $(x,y)$ whenever $y$ is a
best match of $x$, cf.\ Fig.~\ref{fig:tree_BMG_net_BMG}.  Best match graphs
derived from trees in this manner have been studied in detail in
\cite{Geiss2019,Geiss2020}.  In particular, a vertex-colored digraph
$(G,\sigma)$ is a best match graph derived from a tree if and only if it is
color-sink-free (i.e., every vertex has an out-neighbor of each color), and
the pair of sets of rooted and forbidden triples obtained from induced
subgraphs on three vertices is consistent
\cite{Schaller2021,SCHALLER202163}. Moreover, there is a unique least
resolved tree explaining $(G,\sigma)$ that can be constructed in polynomial
time from these triples.

However, evolutionary histories are not always described adequately by
trees. Processes such as hybridization or horizontal gene transfer create
reticulate patterns of ancestry, in which lineages may merge as well as
split. Thus, there has been a growing interest in mathematical phylogenetics
to generalize trees to ``networks'', i.e., directed acyclic graphs (DAG)
$N=(V,E)$ with a unique root $\rho$
\cite{Huson:2010,Kong:22,SolisLemus:26}.  Denoting by $\preceq$ the partial
order on $V$, i.e., $x\preceq y$ if $y$ lies along a directed path from the
root $\rho$ to $x$, the \emph{leaves} of $N$ are the $\preceq$-minimal
elements, or -- equivalently -- the vertices in $N$ with out-degree
zero. We denote by $X$ the set of leaves and say that $N$ is a network on
$X$.

A key difference between trees and networks is that a least common ancestor
needs no longer be unique, cf.\ Fig.~\ref{fig:tree_BMG_net_BMG}.  Let
$\LCA(x,y)$ comprise the $\preceq$-minimal common ancestors of $x$ and
$y$. As in the tree setting, the leaves $x\in X$ are colored by a map
$\sigma$.  In order to define best matches for a leaf-colored network
$(N,\sigma)$, we consider for each $x\in X$ and each color
$s\in\sigma(X) \coloneqq \{\sigma(x) \mid x \in X\}$ the set
\begin{equation} 
  Q(x,s)\coloneqq \min\nolimits_{\preceq}\{v\in\LCA(x,y)\mid y\in X,\
  \sigma(y)=s\}
\end{equation}
comprising the minimal common ancestors obtained by pairing $x$ with any
leaf of color $s$. If $N$ is a rooted tree, $Q(x,s)$ contains the unique
minimal least common ancestor between $x$ and any vertex $y$ with color
$\sigma(y)=s$. For general networks, $Q(x,s)\subseteq V$ may contain
several mutually incomparable vertices.  This leads to the following
definition, illustrated in Fig.~\ref{fig:tree_BMG_net_BMG}:
\begin{definition}
  A leaf $y\in X$ in the leaf-colored network $(N,\sigma)$ on $X$ is
  a \emph{best match} of $x\in X$, if $\sigma(x)\neq\sigma(y)$ and
  $\LCA(x,y)\subseteq Q(x,\sigma(y))$. The \emph{best match graph}
  $\GG(N,\sigma)$ of $(N,\sigma)$ is the colored digraph with vertex set
  $X$, coloring $\sigma$, and an arc $(x,y)$ if and only if $y$ is a best
  match of $x$.
\end{definition}
In other words, a gene $y$ is a best match of $x$ if there is no gene $y'$
of the same color as $y$ and no $u\in\LCA(x,y')$ such that $u\prec v$ for
some $v\in\LCA(x,y)$.  By construction, every BMG is properly colored. A
colored digraph $(G,\sigma)$ is a BMG if there exists a leaf-colored
network $(N,\sigma)$ with leaf set $V(G)$ such that
$(G,\sigma)=\GG(N,\sigma)$.  In this case, we say that $(N,\sigma)$
\emph{explains} $(G,\sigma)$.

\begin{figure}[t]
    \centering
    \includegraphics[width=0.8\textwidth]{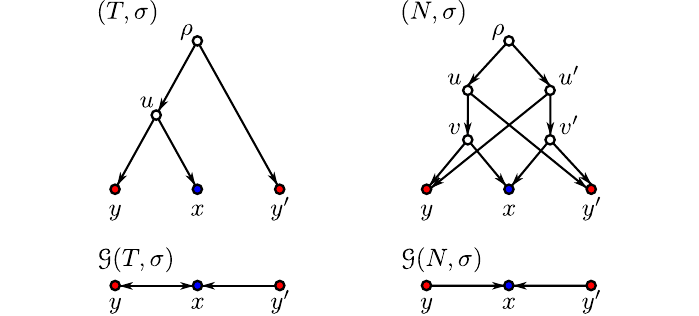}
    \caption{BMGs in trees and networks.  \textbf{Left:} A leaf-colored tree
      $(T,\sigma)$ on $X = \{x,y,y'\}$ and its best match graph
      $\GG(T,\sigma)$ is shown. Since
      $u = \lca(x,y) \prec \lca(x,y') = \rho$, the leaf $y$ is a best match
      of $x$.  Moreover, $x$ is a best match of both $y$ and $y'$.
      \textbf{Right:} For the leaf-colored network $(N,\sigma)$ on $X$, we
      have $\LCA(x,y) = \{v,u'\}$ and $\LCA(x,y') = \{u,v'\}$.  Since
      $Q(x,\sigma(y)) = \{v,v'\}$, neither $y$ nor $y'$ is a best match of
      $x$.  Moreover, $Q(y,\sigma(x)) = \LCA(y,x)$ and
      $Q(y',\sigma(x)) = \LCA(y',x)$ implies that $x$ is a best match of
      both $y$ and $y'$. The BMG $\GG(N,\sigma)$ is shown below
      $(N,\sigma)$. According to the results in \cite{Geiss2019,SCHALLER202163}, $\GG(N,\sigma)$ cannot be explained by a leaf-colored tree.} 
    \label{fig:tree_BMG_net_BMG}
\end{figure}

Reciprocal best matches, i.e., pairs $\{x,y\}$ where $y$ is a best match
for $x$ \emph{and} $x$ is a best match of $y$ play an important role for the
identification of orthologous genes, i.e., those that originate from
speciation events, for tree-like gene phylogenies. The corresponding class
of \emph{reciprocal best match graphs} (RBMG) derived from trees has been
studied in some detail in \cite{Geiss2020}. In contrast to BMGs, a complete
characterization of RBMGs of trees remains an open problem.

In this contribution, we characterize exactly those colored digraphs that
can be explained by a rooted network. To this end, we first introduce BOP
networks in Section~\ref{sec:BOP_networks}, which can be modified by local
``expansions'' into networks that explain any given BMG. This construction
leads to a simple linear-time recognition algorithm for BMGs. In the
affirmative case, an explaining leaf-colored network can be constructed in
$O(|V(G)|^2)$ time. From the recognition algorithm, we derive a simple
structural characterization of BMGs, the main result of our contribution,
Theorem~\ref{thm:char-LCA-BMGs} in Section~\ref{sec:BMG}.  Analogous
characterization and recognition results for undirected RBMGs are described
in Section~\ref{sec:RBMG}.  We close with a brief survey of open questions.

\section{BOP Networks} 
\label{sec:BOP_networks}

The starting point for our characterization of BMGs is a class of dense
networks defined below. To avoid trivial edge cases, we will assume
throughout that $|X|\ge 2$ and there are at least two colors
$|\sigma(X)|\ge2$.

\begin{definition}[BOP network]
  \label{def:pair-lcaN}
  Let $X$ be a finite set equipped with a coloring $\sigma$.  A
  \emph{bichromatic ordered-pair network (BOP network) on $X$} is a
  leaf-colored network $(N,\sigma)$ on $X$ with root $\rho$ constructed as
  follows.  For every ordered pair $(x,y)$ with $x,y\in X$ and
  $\sigma(x)\neq\sigma(y)$, introduce a distinct inner vertex $p_{x,y}$ and
  the arcs $(\rho,p_{x,y})$, $(p_{x,y},x)$, and $(p_{x,y},y)$.  There are
  no further vertices or arcs.
\end{definition}
Note that $p_{x,y}$ and $p_{y,x}$ are distinct if $x\neq y$.  Moreover,
since $|\sigma(X)|\ge 2$, every leaf $x\in X$ occurs below at least one
vertex $p_{x,y}$.  Hence, the digraph in Definition~\ref{def:pair-lcaN} is
indeed a network on $X$, see
Fig.~\hyperref[fig:BOP_extension_Alg_example]{\ref*{fig:BOP_extension_Alg_example}(a)}.

\begin{figure}
  \centering
  \includegraphics[width=\linewidth]{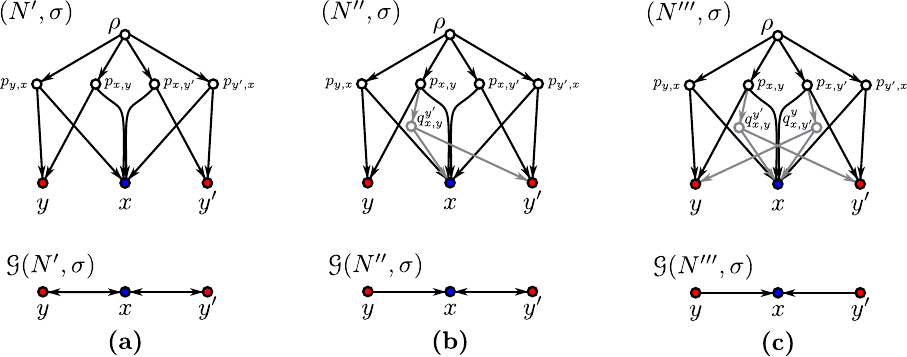}
  \caption{Extensions of BOP networks to explain a BMG.  \textbf{(a)} The
    BOP network $(N',\sigma)$ on $X = \{x,y,y'\}$ is depicted together with
    $\GG(N',\sigma)$.  \textbf{(b)} The leaf-colored network
    $(N'', \sigma)$ coincides with the extension of $N'$ by $[x,y :
    x,y']$. The BMG $\GG(N'',\sigma)$ differs from $\GG(N',\sigma)$ only by
    the absence of the single arc $(x,y)$.  \textbf{(c)} The leaf-colored
    network $(N''',\sigma)$ coincides with the extension of $N''$ by
    $[x,y':x,y]$.  \newline In all examples, the respective network
    $(N,\sigma)$ is the network returned by
    Algorithm~\ref{alg:construction} with input $\GG(N,\sigma)$ for
    $N\in \{N',N'',N'''\}$.}
    \label{fig:BOP_extension_Alg_example}
\end{figure}
  
\begin{lemma}
  \label{lemma:pair-lca-ALL-arcs}
  Let $(N,\sigma)$ be a BOP network on $X$.  Then, for all $x,y\in X$ with
  $\sigma(x)\neq\sigma(y)$, we have $\LCA(x,y)=\{p_{x,y},p_{y,x}\}$.
  Moreover, both $(x,y)$ and $(y,x)$ are arcs of $\GG(N,\sigma)$.
\end{lemma}
\begin{proof}
  Let $(N,\sigma)$ be a BOP network on $X$ and let $x,y\in X$ with
  $\sigma(x)\neq\sigma(y)$.  By construction, the only common ancestors of
  $x$ and $y$ other than the root $\rho$ are precisely $p_{x,y}$ and
  $p_{y,x}$.  These two vertices are $\preceq$-incomparable and both are
  proper descendants of $\rho$.  Hence, $\LCA(x,y)=\{p_{x,y},p_{y,x}\}$.

  Now let $y'\in X$ with $\sigma(y')=\sigma(y)$.  Since
  $\sigma(x)\neq\sigma(y')$, we likewise have
  $\LCA(x,y')=\{p_{x,y'},p_{y',x}\}$.  Every vertex in $\LCA(x,y')$ is
  $\preceq$-incomparable with every vertex in $\LCA(x,y)$.  Hence, for
  every $v\in\LCA(x,y)$, there is no $u\in\LCA(x,y')$ with $u\prec v$.
  Thus, $\LCA(x,y)\subseteq Q(x,\sigma(y))$, and
  $(x,y)\in E(\GG(N,\sigma))$. By symmetry, $(y,x)\in E(\GG(N,\sigma))$ as
  well. \qed
\end{proof}

The BOP network therefore explains the colored digraph containing every possible
arc between distinctly colored vertices of $X$.  In order to explain an
arbitrary BMG $(G,\sigma)$, we have to modify the BOP network in response
to the non-arcs in $(G,\sigma)$: 

\begin{definition}[extension]
    \label{def:B1-extension}
    Let $(N,\sigma)$ be a leaf-colored network on $X$ that contains the BOP
    network on $X$ as a subgraph. Let $x,y,z\in X$ be pairwise distinct
    leaves satisfying $\sigma(x)\neq\sigma(y)=\sigma(z)$. The \emph{extension
    of $N$ by $[x,y:x,z]$} is obtained by introducing a new vertex
    $q^{z}_{x,y}$ and adding the arcs $(p_{x,y},q^{z}_{x,y})$,
    $(q^{z}_{x,y},x)$, and $(q^{z}_{x,y},z)$. We call $(x,y)$ the \emph{first
    part} and $(x,z)$ the \emph{second part} of the extension.
\end{definition}

The notation $q^{z}_{x,y}$ uniquely specifies the extension by which the
vertex is introduced. An extension only adds a new vertex and arcs and does
not delete or modify any part of the underlying BOP network. Hence, the
original BOP network remains a subgraph after every sequence of extensions.
In particular, since $\sigma(x)\neq\sigma(y)$, the vertex $p_{x,y}$
remains present throughout the construction, and extensions can therefore be
applied successively in arbitrary order.

An extension $[x,y:x,z]$ has a simple effect on the LCA structure:
$q^{z}_{x,y}$ becomes an additional LCA of $x$ and $z$ and satisfies
$ q^{z}_{x,y}\prec p_{x,y} $.  Thus, the extension creates an LCA of $x$
and $z$ strictly below one of the LCAs of $x$ and $y$. Thus, $y$ can no
longer be a best match of $x$, see
Fig.~\hyperref[fig:BOP_extension_Alg_example]{\ref*{fig:BOP_extension_Alg_example}(b).}
We next record the structural properties of networks obtained from a BOP
network by successive extensions.

\begin{lemma}
  \label{lem:properties_extensions_BIC_cherry_net}
  Let $(N,\sigma)$ be a BOP network on $X$ with root $\rho$, and let $N'$
  be obtained from $N$ by a finite sequence of extensions such that no
  ordered pair $(x,y)$ occurs as the first part of more than one
  extension. Then the following holds:
  \begin{enumerate}
  \item[(i)] $N'$ is a network on $X$ with root $\rho$.
  \item[(ii)] Let $x,y\in X$ with $\sigma(x)\neq\sigma(y)$.  Then
    $p_{x,y},p_{y,x}\in\LCA_{N'}(x,y)$.  Moreover, if neither $(x,y)$ nor
    $(y,x)$ occurs as the second part of an extension, then
    $\LCA_{N'}(x,y)=\{p_{x,y},p_{y,x}\}$.
  \item[(iii)] If $[x,y:x,z]$ is applied to $N$ then,     
  $q^{z}_{x,y}\in\LCA_{N'}(x,z)$ and $q^{z}_{x,y}\prec_{N'}p_{x,y}$.
  \item[(iv)] Let $x,y,z\in X$ be pairwise distinct with
    $\sigma(x)\neq\sigma(y)=\sigma(z)$.  If there exist
    $u\in\LCA_{N'}(x,z)$ and $v\in\LCA_{N'}(x,y)\setminus\{\rho\}$ such
    that $u\prec_{N'}v$, then the extension $[x,y:x,z]$ is applied in the
    construction of $N'$.
\end{enumerate}
\end{lemma}
\begin{proof}
  (i) Every vertex of the original BOP network $(N,\sigma)$ is either the root
  $\rho$, a vertex $p_{a,b}$, or a leaf. An extension $[a,b:a,c]$
  introduces precisely one additional vertex $q^c_{a,b}$ below $p_{a,b}$,
  with children $a$ and $c$. Since no ordered pair $(a,b)$ occurs as the
  first part of more than one extension, every vertex $p_{a,b}$ has at most
  one non-leaf child introduced by an extension. In particular, no
  extension introduces an arc into the root, into an original vertex
  $p_{a,b}$, or into a previously introduced $q$-vertex. Thus, no directed
  cycle is created. Moreover, every inner vertex remains reachable from
  $\rho$, and every leaf remains a leaf. Hence, $N'$ is a network on $X$
  with root $\rho$ and (i) holds.

  In the following, we write $\LCA_N$ and $\LCA_{N'}$ to distinguish the LCA
  sets of the original and the edited network, respectively.

  (ii) Let $x,y\in X$ with $\sigma(x)\neq\sigma(y)$. By
  Lemma~\ref{lemma:pair-lca-ALL-arcs}, we have
  $\LCA_N(x,y) = \{p_{x,y},p_{y,x}\}$. The vertices $p_{x,y}$ and $p_{y,x}$
  remain common ancestors of $x$ and $y$ throughout the construction. We
  claim that neither of them ceases to be an LCA. Consider $p_{x,y}$. The
  only possible new vertex below $p_{x,y}$ is a vertex introduced by an
  extension whose first part is $(x,y)$. Such a vertex has children $x$ and
  some vertex $z$ with $z\neq y$. In particular, it is not an ancestor of
  $y$. Hence, no new common ancestor of $x$ and $y$ is introduced strictly
  below $p_{x,y}$. Thus, $p_{x,y}\in\LCA_{N'}(x,y)$. The same argument
  applies to $p_{y,x}$. Now suppose that neither $(x,y)$ nor $(y,x)$ occurs
  as the second part of an extension. A new LCA of $x$ and $y$ can only be
  introduced as a $q$-vertex whose two leaf descendants are precisely $x$
  and $y$. By the definition of an extension, this occurs exactly when
  either $(x,y)$ or $(y,x)$ appears as its second part.  Hence,
  $\LCA_{N'}(x,y)=\{p_{x,y},p_{y,x}\}$. This proves (ii).

  (iii) Suppose that the extension $[x,y:x,z]$ is applied. By
  construction, the new vertex $q^{z}_{x,y}$ has only $x$ and $z$ as
  children. Therefore, $q^{z}_{x,y}\in\LCA_{N'}(x,z)$. Moreover, since
  $(p_{x,y},q^{z}_{x,y})$ is an arc, we have
  $q^{z}_{x,y}\prec_{N'}p_{x,y}$. Thus, (iii) holds.

  (iv) Let $x,y,z\in X$ be pairwise distinct with
  $\sigma(x)\neq\sigma(y)=\sigma(z)$, and suppose that there exist
  $u\in\LCA_{N'}(x,z)$ and $v\in\LCA_{N'}(x,y)\setminus\{\rho\}$ such that
  $u\prec_{N'}v$.  Note first that $u$ and $v$ cannot be leaves.  Moreover,
  every non-root inner vertex of $N'$ is either an original vertex
  $p_{a,b}$ or a vertex $q^c_{a,b}$ introduced by an extension $[a,b:a,c]$.
  Since no inner vertex is a proper descendant of a $q$-vertex and
  $u\prec_{N'}v$, the vertex $v$ cannot be a $q$-vertex.  Hence,
  $v=p_{a,b}$ for some $a,b\in X$ with $\sigma(a)\neq\sigma(b)$.  The only
  proper non-leaf descendant of $p_{a,b}$ is, if present, a vertex
  $q_{a,b}^c$ introduced by an extension $[a,b:a,c]$, where
  $\sigma(b)=\sigma(c)$.  Thus, $u = q_{a,b}^c$. Since $q_{a,b}^c$ has
  precisely $a$ and $c$ as leaf descendants and
  $u=q_{a,b}^c\in\LCA_{N'}(x,z)$, we have $\{x,z\}=\{a,c\}$.  Since
  $p_{a,b}\in\LCA_{N'}(x,y)$ and its leaf descendants are $a,b,c$, the
  pairwise distinctness of $x,y,z$ now implies $b=y$. These two
  arguments together with $\sigma(c) = \sigma(b) = \sigma(y) = \sigma(z)$
  imply that $z = b$ and $x = a$.  Thus, the extension $[x,y : x,z]$ was
  applied and (iv) holds.\qed
\end{proof}

\section{Characterization of BMGs}
\label{sec:BMG}

In this section, we show that proper coloring together with the following
condition characterizes BMGs. The sufficiency is established constructively
by an algorithm that produces an explaining network.

An element $x\in X$ is a \emph{single color representative (sicor)} if $x$
is the unique element of its color, that is $\sigma(x) \neq \sigma(y)$ for
all $y \in X \setminus \{x\}$.  A colored digraph $(G,\sigma)$ is said to
have the \emph{sicor-in-hub property} if, for every sicor $x\in X$ and
every $y\in X\setminus\{x\}$, it holds that $(y,x)\in E$. 

\begin{lemma}
  \label{lemma:proper_color_sicor_in_hub_necessary}
  If the colored digraph $(G,\sigma)$ is a BMG, then it is properly colored
  and has the sicor-in-hub property.
\end{lemma}
\begin{proof}
  Proper coloring follows directly from the definition. Let $y\in X$ be a
  sicor and let $x\in X\setminus\{y\}$. Since $y$ is the unique vertex of
  color $\sigma(y)$, we have $Q(x,\sigma(y))=\min_{\preceq}\LCA(x,y)$. As
  every element of $\LCA(x,y)$ is $\preceq$-minimal among the common
  ancestors of $x$ and $y$, it follows that
  $Q(x,\sigma(y))=\LCA(x,y)$. Hence, $\LCA(x,y)\subseteq Q(x,\sigma(y))$,
  and $y$ is a best match of $x$. Thus, $(x,y)\in E$. Since this holds for
  every sicor $y$ and every $x\in X\setminus\{y\}$, $(G,\sigma)$ has the
  sicor-in-hub property. \qed 
\end{proof} 

To show that these two conditions are also sufficient, we construct a
leaf-colored network that explains every properly colored digraph with the
sicor-in-hub property. To this end, we use the class of BOP networks and
extensions to construct an explaining network, see
Fig.~\hyperref[fig:tree_BMG_net_BMG]{\ref*{fig:tree_BMG_net_BMG}(c)}.

\begin{algorithm}[ht]
\caption{Construction of a network explaining a BMG}
\label{alg:construction}
\begin{algorithmic}[1]
  \Require A properly colored digraph $(G=(X,E),\sigma)$ with the
           sicor-in-hub property
  \Ensure A leaf-colored network $(N,\sigma)$ explaining $(G,\sigma)$
  \State Let $(N,\sigma)$ be the BOP network on $X$
  \ForAll{$x,y\in X$ with $\sigma(x)\neq\sigma(y)$ and $(x,y)\notin E$}
     \State Choose $y'\in X\setminus\{y\}$ with $\sigma(y')=\sigma(y)$
        \label{alg1:line:extend}
     \State Extend $N$ by $[x,y:x,y']$
  \EndFor
  \State \Return $(N,\sigma)$
\end{algorithmic}
\end{algorithm}

\begin{lemma}
\label{lemma:constructioncorrect}
Let $(G,\sigma)$ be a properly colored digraph with the sicor-in-hub
property. Then Algorithm~\ref{alg:construction} returns a leaf-colored
network $(N,\sigma)$ explaining $(G,\sigma)$.
\end{lemma}
\begin{proof} 
  Let $(G=(X,E),\sigma)$ be properly colored and satisfy the sicor-in-hub
  property, and let $(N,\sigma)$ be the network returned by
  Algorithm~\ref{alg:construction}.  For every non-arc $(x,y)$ with
  $\sigma(x)\neq\sigma(y)$, the algorithm applies exactly one extension
  whose first part is $(x,y)$. Hence, no ordered pair occurs as the first
  part of more than one extension. By
  Lemma~\ref{lem:properties_extensions_BIC_cherry_net}(i), the returned
  digraph $N$ is therefore a network on $X$. Note that the choice of $y'$
  in Line~\ref{alg1:line:extend} in Algorithm~\ref{alg:construction} is
  always possible, since if $(x,y)\notin E$ and $y$ were a sicor, then the
  sicor-in-hub property would imply $(x,y)\in E$, a contradiction. Hence,
  if $(x,y)\notin E$, then there exists another vertex $y'\neq y$ with
  $\sigma(y')=\sigma(y)$.

  Next we show that $\GG(N,\sigma) = (G,\sigma)$ holds.  By definition, the
  vertex sets coincide.  Let $\widetilde E$ denote the arc set of
  $\GG(N,\sigma)$.  It remains to prove $E=\widetilde E$.  Let $x,y\in X$
  be distinct. If $\sigma(x)=\sigma(y)$, then $(x,y)\notin E$, as
  $(G,\sigma)$ is properly colored and $(x,y)\notin\widetilde E$ by the
  definition of best matches. Hence, assume in the following that
  $\sigma(x)\neq\sigma(y)$.  Suppose $(x,y)\notin E$. As observed above,
  $y$ is not a sicor. Hence, Algorithm~\ref{alg:construction} chooses some
  $y'\neq y$ with $\sigma(y')=\sigma(y)$ and applies the extension
  $[x,y:x,y']$. By Lemma~\ref{lem:properties_extensions_BIC_cherry_net},
  $p_{x,y}\in\LCA(x,y)$ and $q^{y'}_{x,y}\in\LCA(x,y')$, and moreover
  $q^{y'}_{x,y}\prec p_{x,y}$. Thus, at least one element of $\LCA(x,y)$ is
  not $\preceq$-minimal among all LCAs of $x$ and vertices of color
  $\sigma(y)$. Consequently, $\LCA(x,y)\not\subseteq Q(x,\sigma(y))$, and
  $y$ is not a best match of $x$. Hence, $(x,y)\notin\widetilde E$.

  Conversely, suppose $(x,y)\notin\widetilde E$. Then $y$ is not a best
  match of $x$, and therefore $\LCA(x,y)\not\subseteq
  Q(x,\sigma(y))$. Hence, there exists some $v\in\LCA(x,y)$ that is not
  $\preceq$-minimal among all LCAs of $x$ and leaves of color
  $\sigma(y)$. Consequently, there exists a vertex $y'\in X$ with
  $\sigma(y')=\sigma(y)$ and some $u\in\LCA(x,y')$ such that
  $u\prec v$. Necessarily, $y'\neq y$, since distinct LCAs of the same
  pair cannot be strictly comparable. By
  Lemma~\ref{lem:properties_extensions_BIC_cherry_net}(ii),
  $p_{x,y}\in\LCA(x,y)$.  Hence, $\rho\notin\LCA(x,y)$ and therefore
  $v\in\LCA(x,y)\setminus\{\rho\}$. Since
  $\sigma(x)\neq\sigma(y)=\sigma(y')$, all assumptions of
  Lemma~\ref{lem:properties_extensions_BIC_cherry_net}(iv) are
  satisfied. Thus, the extension $[x,y:x,y']$ was applied during
  Algorithm~\ref{alg:construction}. Such an extension is applied only if
  $(x,y)\notin E$. Hence, $(x,y)\notin\widetilde E$ implies
  $(x,y)\notin E$. We conclude that $E=\widetilde E$, and therefore
  $\GG(N,\sigma)=(G,\sigma)$. Thus, $(N,\sigma)$ explains $(G,\sigma)$.\qed
\end{proof}

\begin{lemma}
  \label{lemma:Alg1-runningtime}
  Algorithm~\ref{alg:construction} can be implemented to run in
  $O(|V(G)|^2)$ time provided the input graph $(G,\sigma)$ is a properly
  colored digraph with the sicor-in-hub property.
\end{lemma} 
\begin{proof}
  Let $n=|X|$. In the construction of the BOP network on $X$, we add a constant number of vertices
  and arcs for every ordered pair $x,y\in X$ with $\sigma(x)\neq\sigma(y)$. Since there are at most
  $n^2$ such pairs, the initial BOP network can be constructed in $O(n^2)$ time. In a preprocessing
  step, $X_s = \{x \in X \mid \sigma(x) = s\}$ is stored as a list together with its cardinality for
  each $s \in \sigma(X)$. If $|X_s|\ne 1$, two distinct representatives of $X_s$ are stored. All
  color classes and the required auxiliary information can be constructed in $O(n)$ time, while
  different representations of $G$ can be interconverted in $O(n^2)$ time. Now fix $x,y\in X$ and
  consider the \emph{for}-loop. Assuming that $G$ is given as an adjacency matrix, verifying whether
  $(x,y)\notin E$ is clearly a constant-time operation. Likewise, it can be checked in constant time
  whether or not $\sigma(x)\neq\sigma(y)$. Whenever an extension is required, $y$ is not a sicor and
  hence $|X_{\sigma(y)}|\geq 2$. Therefore, a vertex $y'\in X_{\sigma(y)}\setminus\{y\}$ can be
  selected in constant time by choosing one of the two stored representatives that is distinct from
  $y$. This is equivalent to selecting a vertex $y' \in X \setminus \{y\}$ with $\sigma(y') =
  \sigma(y)$, as required in Line~\ref{alg1:line:extend}. Moreover, the extension $[x,y:x,y']$ can
  clearly be carried out in constant time. Since the \emph{for}-loop is executed at most $n^2$ times
  and each iteration, including the work in Line~\ref{alg1:line:extend}, can be completed in
  constant time, the total running time of the \emph{for}-loop is $O(n^2)$. In summary,
  Algorithm~\ref{alg:construction} has time complexity $O(n^2)$. \qed
\end{proof}

We are now in the position to state the main result of this contribution.
\begin{theorem}
  \label{thm:char-LCA-BMGs}
  A colored digraph $(G,\sigma)$ is a BMG if and only if $(G,\sigma)$ is
  properly colored and has the sicor-in-hub property.  Moreover, it can be
  decided in $O(|V(G)|+|E(G)|)$ time whether a given colored digraph
  $(G,\sigma)$ is a BMG. In the affirmative case, a leaf-colored network
  explaining $(G,\sigma)$ can be constructed in $O(|V(G)|^2)$ time.
\end{theorem}
\begin{proof}
  Lemma~\ref{lemma:proper_color_sicor_in_hub_necessary} and
  \ref{lemma:constructioncorrect} directly imply that $(G,\sigma)$ is a BMG
  if and only if $(G,\sigma)$ is properly colored and has the sicor-in-hub
  property. In order to estimate the running time, we set
  $n\coloneqq |V(G)|$ and $m\coloneqq |E(G)|$ and assume that $G$ is given
  by adjacency lists. We group the vertices into color classes
  $X_s \coloneqq \{x \in X \mid \sigma(x) = s\}$ and record the
  cardinalities $|X_s|$ in $O(n)$ time. We verify that
  $\sigma(x)\neq\sigma(y)$ for all $(x,y)\in E(G)$ and compute the
  in-degrees of all vertices in $O(m)$ total time.  A vertex is a sicor if
  $|X_{\sigma(x)}|=1$ and, once proper coloring has been verified, the
  sicor-in-hub property is equivalent to $\indeg_G(x)=n-1$ for every sicor
  $x$. Given the preprocessing, these conditions can be checked in constant
  time for each $x\in X$, amounting to $O(n)$ total time. Thus, we can
  decide in $O(n+m)$ time whether a colored digraph is a BMG. In the
  affirmative case, Algorithm~\ref{alg:construction} constructs an
  explaining network in $O(n^2)$ time by
  Lemma~\ref{lemma:Alg1-runningtime}.\qed
\end{proof}

Theorem~\ref{thm:char-LCA-BMGs} gives rise to a simple characterization of
the description of all possible BMGs for a given set of colored vertices:
\begin{proposition}
  \label{prop:BMG-interval}
  Let $(G=(X,E),\sigma)$ be a colored digraph. Then $(G,\sigma)$ is a BMG
  if and only if $E_{\min}\subseteq E\subseteq E_{\max}$, where
  $E_{\max} \coloneqq \{(x,y)\in X\times X \mid \sigma(x)\neq\sigma(y)\}$
  and
  $E_{\min} \coloneqq \{(x,y)\in X\times X \mid y \text{ is a sicor and }
  x\neq y\}$.
\end{proposition}
\begin{proof}
By Theorem~\ref{thm:char-LCA-BMGs}, $(G,\sigma)$ is a BMG if and only if
it is properly colored and satisfies the sicor-in-hub property.
Proper coloring is equivalent to $E\subseteq E_{\max}$, 
while the sicor-in-hub property holds precisely if
$E_{\min}\subseteq E$. \qed
\end{proof}

As an immediate consequence of Proposition~\ref{prop:BMG-interval}, we have
\begin{corollary}
  If $(G_1,\sigma)$ and $(G_2,\sigma)$ are BMGs on the same colored vertex
  set, then the colored digraphs with arc sets $E(G_1)\cap E(G_2)$ and
  $E(G_1)\cup E(G_2) $ are BMGs as well.
\end{corollary}

\begin{corollary}
  If the set of colored vertices $(X,\sigma)$ has no sicors, then every
  properly colored digraph on $(X,\sigma)$ is a BMG. If every vertex in $X$
  is a sicor, then there is exactly one BMG on $(X,\sigma)$, namely the
  digraph containing both arcs $(x,y)$ and $(y,x)$ for every pair of
  distinct vertices $x,y\in X$.
\end{corollary}
\begin{proof}
This follows from Proposition~\ref{prop:BMG-interval} and the facts that, 
if there are no sicors, then $E_{\min}=\emptyset$, and 
if every vertex is a sicor, then $E_{\min}=E_{\max}$. \qed
\end{proof}

The latter result allows us to count the number of BMGs with a fixed
coloring.
\begin{corollary}
\label{cor:number-BMGs-fixed-coloring}
Let $(X,\sigma)$ be a fixed set of colored vertices with $n=|X|$ and $h$
sicors, and let
$M \coloneqq \left| \left\{ \{x,y\}\subseteq X \;\middle|\;
    \sigma(x)\neq\sigma(y) \right\} \right|$. Then there are
$2^{2M-h(n-1)}$ BMGs on $(X,\sigma)$. 
\end{corollary}
\begin{proof}
  There are exactly $2M$ possible arcs between differently colored
  vertices, and hence $ |E_{\max}|=2M$. For every sicor $y$, the forced
  arcs are precisely $F_y\coloneqq\{(x,y)\mid x\in X\setminus\{y\}\}$ and
  hence $|F_y|=n-1$. Moreover, if $y$ and $z$ are distinct sicors, then
  $F_y\cap F_z=\emptyset$ must hold.  Consequently, $|E_{\min}|=h(n-1)$.
  By Proposition~\ref{prop:BMG-interval}, every subset of
  $E_{\max}\setminus E_{\min}$ can be chosen independently, resulting in
  $ 2^{|E_{\max}\setminus E_{\min}|} = 2^{\,2M-h(n-1)} $ distinct BMGs.
  \qed
\end{proof}

Theorem~\ref{thm:char-LCA-BMGs} furthermore implies:
\begin{corollary}
  \label{cor:two-color-BMG}
  Let $(G=(X,E),\sigma)$ be a colored digraph. Set 
  $X_s \coloneqq \{x \in X \mid \sigma(x) = s\}$ and denote by $G_{st}$
  the subgraph of $G$ induced by $X_s \cup X_t$ and let $\sigma_{st}$ be
  the restriction of $\sigma$ to $X_s \cup X_t$. Then $(G,\sigma)$ is a BMG
  if and only if, for all distinct colors $s,t\in\sigma(X)$, the
  induced colored subgraph $(G_{st},\sigma_{st})$ is a BMG.
\end{corollary}
\begin{proof}
  Suppose first that $(G,\sigma)$ is a BMG. By
  Theorem~\ref{thm:char-LCA-BMGs}, $(G,\sigma)$ is properly colored and has
  the sicor-in-hub property. Let $s,t\in\sigma(X)$ be distinct. Clearly,
  $(G_{st},\sigma_{st})$ is properly colored. Moreover, a vertex $y\in X_t$
  is a sicor in $G_{st}$ if and only if $y$ is a sicor of color $t$ in
  $G$. Since $(G,\sigma)$ has the sicor-in-hub property, $(x,y)\in E(G)$
  for every $x\in X_s$. Thus, $(x,y)\in E(G_{st})$ for every $x\in
  X_s$. The same argument applies with $s$ and $t$ interchanged. Hence,
  $(G_{st},\sigma_{st})$ has the sicor-in-hub property and is therefore a
  BMG.
  
  Conversely, suppose that $(G_{st},\sigma_{st})$ is a BMG for every pair
  of distinct colors $s,t$. By Theorem~\ref{thm:char-LCA-BMGs}, every
  $(G_{st},\sigma_{st})$ is properly colored and has the sicor-in-hub
  property.  Hence, $(G,\sigma)$ is properly colored.  Now let $y\in X$ be
  a sicor with color $t$ and let $x\in X\setminus\{y\}$ with color
  $s\neq t$.  Since $X_t=\{y\}$, the vertex $y$ is also a sicor in
  $G_{st}$. The sicor-in-hub property of $(G_{st},\sigma_{st})$ implies
  $(x,y)\in E(G_{st})$, and hence $(x,y)\in E(G)$. Thus, $(G,\sigma)$ has
  the sicor-in-hub property.  Application of
  Theorem~\ref{thm:char-LCA-BMGs} shows that $(G,\sigma)$ is a BMG.
  \qed
\end{proof}

Corollary~\ref{cor:two-color-BMG} should not be mistaken for closure
under arbitrary induced subgraphs. In fact, the class of BMGs is not
hereditary with respect to vertex deletion. To see this, consider two color
classes $X_s=\{a_1,a_2\}$ and $X_t=\{b_1,b_2\}$.  Since neither color class
contains a sicor, every properly colored digraph on these four vertices
satisfies the sicor-in-hub property and is therefore a BMG. Choose such a
digraph with $(b_1,a_1)\notin E$.  After deleting $a_2$, the vertex $a_1$
becomes a sicor, while $(b_1,a_1)$ is still absent. Hence, the resulting
induced colored subgraph does not satisfy the sicor-in-hub property and is
not a BMG.  Thus, BMGs are closed under restriction to unions of entire
color classes, but not under arbitrary vertex-induced subgraphs.

\section{Characterization of RBMGs}
\label{sec:RBMG}

For a colored digraph $(G=(X,E),\sigma)$, let $G^{\leftrightarrow}$ denote
its \emph{reciprocal part}, i.e., the undirected graph with vertex set $X$
and $xy\in E(G^{\leftrightarrow})$ if and only if $(x,y),(y,x)\in E$.  We
call an undirected colored graph $(H,\sigma)$ a \emph{reciprocal best match
  graph (RBMG)} if $H=G^{\leftrightarrow}$ for some BMG $(G,\sigma)$, cf.\
Fig.~\ref{fig:RBMG}.
\begin{figure}
  \centering
  \includegraphics[width=\linewidth]{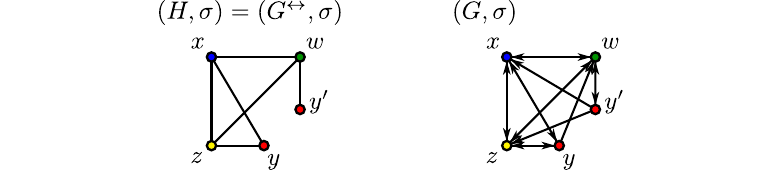}
  \caption{Shown is a RBMG $(H,\sigma)$ and a BMG $(G,\sigma)$ with
    $H = G^{\leftrightarrow}$.  In this example, $w$, $x$, and $z$ are
    sicors and $G$ is constructed from $H$ according to the procedure
    described in the proof of Theorem~\ref{thm:char-RBMG}.}
  \label{fig:RBMG}
\end{figure}

\begin{theorem}
\label{thm:char-RBMG}
The undirected colored graph $(H=(X,F),\sigma)$ is an RBMG if and only if
\begin{enumerate}
\item $(H,\sigma)$ is properly colored, and
\item any two distinct sicors of $(H,\sigma)$ are adjacent in $H$.
\end{enumerate}
Moreover, it can be decided in $O(|X|+|F|)$ time whether $(H,\sigma)$ is an
RBMG.
\end{theorem}
\begin{proof}
  Suppose first that $(H,\sigma)$ is an RBMG. Then there exists a BMG
  $(G=(X,E),\sigma)$ such that $H=G^{\leftrightarrow}$. Since $(G,\sigma)$
  is properly colored by Theorem~\ref{thm:char-LCA-BMGs}, so is
  $(H,\sigma)$.  Now let $x,y\in X$ be distinct sicors. Since $y$ is a
  sicor, the sicor-in-hub property of $(G,\sigma)$ implies $(x,y)\in
  E$. Similarly, since $x$ is a sicor, we have $(y,x)\in E$.  Hence,
  $xy\in E(H)$ and every two distinct sicors are adjacent in $H$.

  Conversely, suppose that $(H=(X,F),\sigma)$ is properly colored and that
  any two distinct sicors are adjacent in $H$.  We construct a digraph
  $G=(X,E)$ as follows.  For every edge $xy\in F$, include both arcs
  $(x,y)$ and $(y,x)$ in $E$.  Moreover, for every sicor $y$ and every
  $x\in X\setminus\{y\}$ with $xy\notin F$, include the arc $(x,y)$.  There
  are no further arcs.  Since $H$ is properly colored, all arcs introduced
  in this way join vertices of distinct colors.  Hence, $(G,\sigma)$ is
  properly colored.  By construction, every sicor $y$ receives an arc from
  every $x\neq y$, and therefore $(G,\sigma)$ has the sicor-in-hub
  property.  Theorem~\ref{thm:char-LCA-BMGs} implies that $(G,\sigma)$ is a
  BMG.  It remains to verify that $H=G^{\leftrightarrow}$.  Every edge
  $xy\in F$ gives rise to both arcs $(x,y)$ and $(y,x)$ and therefore
  $xy\in E(G^{\leftrightarrow})$.  Conversely, suppose that $xy\notin F$.
  Hence, both $(x,y)$ and $(y,x)$ could belong to $E$ only if both $x$ and
  $y$ were sicors.  This is impossible, since distinct sicors are adjacent
  in $H$ by assumption.  Thus, at most one of $(x,y)$ and $(y,x)$ belongs
  to $E$, and consequently $xy\notin E(G^{\leftrightarrow})$.  Therefore,
  $H=G^{\leftrightarrow}$, and $(H,\sigma)$ is an RBMG.

  For the recognition algorithm, first determine all color-class sizes by
  scanning the vertices once.  This identifies all sicors in $O(|X|)$
  time. Next, scan every edge $xy\in F$ once. During this scan, verify that
  $\sigma(x)\neq\sigma(y)$ and count the number of edges whose two
  endpoints are sicors.  Thus, proper coloring and the number of edges
  induced by the sicors can be determined in $O(|X|+|F|)$ time. If there
  are $h$ sicors, they induce a clique if and only if the number of edges
  between them is $\binom{h}{2}$. Hence, Theorem~\ref{thm:char-RBMG} yields
  an $O(|X|+|F|)$-time recognition algorithm.
  \qed
\end{proof}

The proof of Theorem~\ref{thm:char-RBMG} is constructive.
In particular, every RBMG $(H,\sigma)$ admits a BMG $(G,\sigma)$ with
$H=G^{\leftrightarrow}$ that can be constructed in $O(|X|^2)$ time.
We conclude our discussion of RBMGs with two simple consequences.

\begin{corollary}
  \label{cor:RBMG-1}
  Let $(H=(X,F),\sigma)$ be an undirected colored graph.  Then $(H,\sigma)$
  is an RBMG if and only if $F_{\min}\subseteq F\subseteq F_{\max}$, where
  $F_{\max}\coloneqq\{\{x,y\}\subseteq X \mid \sigma(x)\neq\sigma(y)\}$ and
  $F_{\min}\coloneqq\{\{x,y\}\subseteq X\mid x\text{ and }y\text{ are
    distinct sicors}\}$.
\end{corollary}
\begin{proof}
  By Theorem~\ref{thm:char-RBMG}, $(H,\sigma)$ is an RBMG if and only if it
  is properly colored and the sicors induce a clique.  Proper coloring is
  equivalent to $F\subseteq F_{\max}$, while the condition that the sicors
  induce a clique is equivalent to $F_{\min}\subseteq F$.  Hence,
  $(H,\sigma)$ is an RBMG if and only if
  $F_{\min}\subseteq F\subseteq F_{\max}$. \qed
\end{proof}

\begin{corollary}
  \label{cor:RBMG-2}
  Let $(X,\sigma)$ be a fixed colored set. Moreover, let
  $M\coloneqq|\{\{x,y\}\subseteq X\mid \sigma(x)\neq\sigma(y)\}|$ and let
  $h$ denote the number of sicors.  The number of RBMGs on the fixed
  colored vertex set $(X,\sigma)$ is $2^{M-\binom{h}{2}}$.
\end{corollary}
\begin{proof}
  There are exactly $M$ possible bichromatic edges, and therefore
  $|F_{\max}|=M$.  Since any two distinct sicors have distinct colors,
  every unordered pair of sicors belongs to $F_{\max}$, and thus
  $|F_{\min}|=\binom{h}{2}$.  Every edge in $F_{\max}\setminus F_{\min}$
  can be chosen independently. The corollary now follows from
  $|F_{\max}\setminus F_{\min}| =M-\binom{h}{2}$.  \qed
\end{proof}

\section{Concluding Remarks and Outlook}

Our characterization of best match graphs shows that, when arbitrary rooted
networks are admitted as explaining structures, the best-match condition
considered here imposes surprisingly few restrictions on the resulting
digraphs. Apart from proper coloring, the only global constraint is induced
by colors represented by a single vertex. Thus, from a graph-theoretic
point of view, the resulting class of BMGs is rather broad.  From an
applications point of view, furthermore, the notion of best matches
employed in this contribution does not imply the existence of a best match
of each gene in every other species. Both considerations suggest to explore
alternative definitions of best matches.

For example, we may define a leaf $y\in X$ as a \emph{weak best match} of
$x\in X$ if $\sigma(x)\neq\sigma(y)$ and
$\LCA(x,y)\cap Q(x,\sigma(y))\neq\emptyset$, i.e., at least one LCA of $x$
and $y$ is $\preceq$-minimal among all the LCAs of $x$ and leaves with the
same color as $y$. It follows immediately from the definition of
$Q(x,s)$ and the fact that $\LCA(x,y)\ne\emptyset$ for all $x,y\in X$
that $x\in X$ has at least one best match in each color class, i.e., the
corresponding ``weak best match graphs'' are color-sink-free. Moreover,
BMGs and their weak counterparts coincide if the explaining networks are
restricted to trees.

It will also be interesting to explore the (weak) BMGs explained by
different classes of networks interpolating between rooted trees studied
in \cite{Geiss2019} and arbitrary rooted networks considered
here. Interesting candidates are in particular classes of networks that
guarantee unique least common ancestors, see
\cite{Hellmuth:23a,Lindeberg2025}. We can expect that such restrictions
on the explaining network will enforce non-trivial structural properties on
the resulting BMGs and likely exhibit a much closer connection between the
combinatorics of the digraph and the topology of its explaining network.

In applications one is most commonly interested in parsimonious
explanations. The construction of explanations for BMGs via BOP networks,
however, produces dense networks also for BMGs that can be explained by
trees. It is therefore natural to ask to what extent an explaining
network can be simplified while preserving the induced BMG. This includes
identifying redundant vertices and arcs, finding minimal or canonical
explaining networks, and determining the computational complexity of
obtaining such explanations. These questions may also provide a more
informative link between the structure of a BMG and that of networks
explaining it.

\subsection*{Acknowledgements}

Research in the Stadler lab is supported by the Federal Ministry of
Research, Technology, and Space of Germany through DAAD project 57616814
(SECAI, School of Embedded Composite AI), jointly with S{\"a}chsisches
Staatsministerium für Wissenschaft, Kultur und Tourismus in the programme
Center of Excellence for AI-research \emph{Center for Scalable Data
Analytics and Artificial Intelligence Dresden/Leipzig} (ScaDS.AI, proj.\
no.\ SCADS24B), and by the Deutsche Forschungsgemeinschaft (DFG, German
Research Foundation) under Germany's Excellence Strategy (EXC-3105/1,
533765739).

\bibliography{B1_netBMG}

%\clearpage

\end{document}